\documentclass[letterpaper, 10 pt, conference]{ieeeconf}  

\UseRawInputEncoding

\IEEEoverridecommandlockouts                              

\usepackage{color}

\usepackage{amsmath} 
\usepackage{amssymb}  
\usepackage{cite}
\usepackage{footmisc}

\usepackage{algorithm}  
\usepackage{algpseudocode}  

\usepackage{booktabs}
\usepackage{makecell}
\usepackage{graphicx}
\usepackage{threeparttable}
\usepackage{cuted}

\newcommand{\bs}[1]{\boldsymbol{#1}}

\newtheorem{definition}{Definition}
\newtheorem{assumption}{Assumption}
\newtheorem{remark}{Remark}
\newtheorem{corollary}{Corollary}

\newtheorem{lemma}{Lemma}
\newtheorem{theorem}{Theorem}

\usepackage{xcolor}
\usepackage{xpatch}

\makeatletter
\ExplSyntaxOn
\cs_new:Npn \bibColoredItems #1#2
  {
    \clist_map_inline:nn {#2} { \cs_new:cpn {bib@colored@##1} {#1} } 
  }
\ExplSyntaxOff

\newcommand\bib@setcolor[1]{%
  \ifcsname bib@colored@#1\endcsname
    \expandafter\color\expandafter{\csname bib@colored@#1\endcsname}
  \else
    \normalcolor
  \fi
}

\xpatchcmd\@bibitem
  {\item}
  {\bib@setcolor{#1}\item}
  {}{\fail}

\xpatchcmd\@lbibitem
  {\item}
  {\bib@setcolor{#2}\item}
  {}{\fail}
\makeatother

\title{\LARGE \bf
On the Existence of Linear Observed Systems \\ on Manifolds with Connection
}

\author{Changwu Liu$^{1}$ and Yuan Shen$^{2}$
\thanks{$^{1}$Changwu Liu is with Department of Electronic Engineering, Tsinghua University, Beijing, China
        {\tt\small liucw\_ee@tsinghua.edu.cn}}%
\thanks{$^{2}$Yuan Shen is with Department of Electronic Engineering, Tsinghua University, Beijing, China
        {\tt\small shenyuan\_ee@tsinghua.edu.cn}}%
}

\begin{document}

\vspace*{\fill}
\begin{strip}
\copyright{ 2024 IEEE.} Personal use of this material is permitted. Permission from IEEE must be obtained for all other uses, in any current or future media, including reprinting/republishing this material for advertising or promotional purposes, creating new collective works, for resale or redistribution to servers or lists, or reuse of any copyrighted component of this work in other works.

This is the Author Accepted Version of: C. Liu and Y. Shen, ``On the Existence of Linear Observed Systems on Manifolds with Connection,'' \emph{IEEE Control Systems Letters}, 2024.

DOI: 10.1109/LCSYS.2024.3506222

\end{strip}
\vspace*{\fill}
\newpage

\maketitle
\thispagestyle{empty}
\pagestyle{empty}

\begin{abstract}

Linear observed systems on manifolds are a special class of nonlinear systems whose state spaces are smooth manifolds but possess properties similar to linear systems. Such properties can be characterized by preintegration and exact linearization with Jacobians independent of the linearization point. Non-biased IMU dynamics in navigation can be constructed into linear observed settings, leading to invariant filters with guaranteed behaviors such as local convergence and consistency. In this letter, we establish linear observed property for systems evolving on a smooth manifold through the connection structure endowed upon this space. Our key findings are the existence of linear observed systems on manifolds poses constraints on the curvature of the state space, beyond requiring the dynamics to be compatible with some connection-preserving transformations. Specifically, the flat connection case reproduces the characterization of linear observed systems on Lie groups, showing our theory is a true generalization.

\end{abstract}


\section{Introduction}

\PARstart{L}{inearity} of a dynamical system can be characterized from many aspects, perhaps the two most important ones are the exactness of autonomous error evolution and the possibility to conduct preintegration\cite{Pre_Int,PreIntMath}. These two properties can be generalized beyond linear systems on Euclidean spaces when the state takes value in a smooth manifold. We term linear observed for systems evolving on manifold state space which admit the above two characteristics. It's possible to design nonlinear observers for linear observed systems on non-Euclidean spaces with guaranteed local convergence\cite{InEKF, LoS} or consistency\cite{ExploitSymm}. For instance, estimation of the attitude, position and velocity of a rigid body in motion using IMU (inertial measurement unit) is a central challenge in the field of navigation. Such problem is naturally formulated in a 9-$\dim$ manifold state space ${\rm SO}(3)\times\mathbb{R}^6$ not diffeomorphic to $\mathbb{R}^9$. To one's first sight, non-biased IMU dynamics is highly nonlinear. However, an ingenious ${\rm SE}_2(3)$ group structure\cite{InEKF, TFG} can be given to the above state space which makes the system linear observed, leading to a huge amount of successful IMU-based navigation algorithms\cite{VI-InEKF,InGVIO,ContactAid,InSmoothingCDC}.

Unlike endowing a global group structure on the state manifold in advance\cite{InEKF,LoS,TFG,AnnuRevInEKF}, we attempt to generalize linear observed properties to systems evolving on arbitrary smooth manifolds and study the conditions for the existence of this class of systems. To rigorously define linear observed properties, we need extra structures on the state manifold besides smoothness to relate different tangent spaces and to move from arbitrary points to their neighborhoods by tangent vectors. Retraction\cite{Opt_Mfd} defined pointwise is not enough for our purpose. A fundamental structure on manifold called connection\cite{DiffGeo} is needed. Only a few papers\cite{KF_mfd} have noted the importance of connection in observer design in the control community. None of them uses such structure for the purpose of studying linear observed properties. The key advantage is that a  connection can be globally established on an arbitrary smooth manifold while a group structure can not be. 

The flow emanating from a point $\bs{x}\in\mathbb{R}^m$ of a linear system on $\mathbb{R}^m$ is the action of some evolving linear isomorphism acting on $\bs{x}$. Connection-preserving, a.k.a. affine, transformation is the generalization of such linear isomorphism. Exact linearization is thus modelled as the compatibility between the system flow and some affine transformation. Preintegrability can be viewed as a self-similarity constraint to the flow field. Linear observed systems in this letter are then defined by the ability to conduct preintegration and exact linearization. The central theoretical result is that the existence of linear observed systems on a manifold endowed with connection requires restrictions to the state space itself, e.g. the curvature, in addition to the special forms that the dynamics should obey. Specifically, the flat Cartan (+)-connection case of our theory reproduces the characterization of linear observed system on Lie groups. Our key contributions are listed herein:
\begin{itemize}
\item we have rigorously generalized the concept of linear observed systems to dynamics evolving on arbitrary smooth manifolds leveraging the connection structure, covering e.g. Riemannian manifold and Lie groups;
\item the existence of linear observed systems on manifolds admitting a connection poses conditions on the state manifold itself, beyond the requirements on the structure of the system flow; specifically, the flat case reproduces characterization of group affine systems;
\item an example of observer design on $S^2$ is provided to reveal the mechanisms of our generalization. 
\end{itemize}

Section~\ref{sec::pre} introduces necessary mathematical preliminaries. Our core theory rolls out in section~\ref{sec::theory}. An example from navigation with state space $S^2$ is provided to reveal the mechanisms of the theory. Section~\ref{sec::conclusion} concludes this letter. All proofs are detailed in the appendix.

\section{Preliminaries}\label{sec::pre}

Mathematical preliminaries and notations
used in this letter are briefly summarized herein. The concepts in this section are taken from \cite{DiffGeo,DGLG,VisualDG}, and readers unfamiliar with those may find such references useful. Let $M$ denote a smooth manifold, a topological space locally diffeomorphic to Euclidean space with smooth transition maps whenever two coordinate charts overlap. For any $p\in M$, let $T_pM$ be the tangent space at $p$ and $TM$ be the tangent bundle over $M$. $T^*_pM$ denotes the dual vector space of $T_pM$. Let $\Gamma(TM)$ denote the vector field over $M$, defined as a smooth cross section of $TM$. $\Gamma(TM)$ is identified as either infinite-dimensional $\mathbb{R}$-vector space or $C^\infty(M)$-module contextually, where $C^\infty(M)$ is the commutative ring of smooth functions on $M$. An $(r,s)$-tensor field $T$ on $M$ is a multi-linear map from $(\Gamma(T^*M))^r\times(\Gamma(TM))^s$ to $C^\infty(M)$ and the collection of such tensors is denoted by $(\otimes^r\Gamma(TM))\otimes(\otimes^t\Gamma(T^*M))$ using the tensor product $\otimes$.

An affine connection $\nabla$ is an extra structure on a manifold other than smoothness. $\nabla$ maps a vector field and an $(r,s)$-tensor field to another $(r,s)$-tensor field with the following properties: (i) $\mathbb{R}$-linearity in the upper slot: $\nabla_X(k_1T+k_2S)=k_1\nabla_XT+k_2\nabla_XS$ (ii) $C^\infty(M)$-linearity in the lower entry: $\nabla_{fX+Y}T=f\nabla_XT+\nabla_YT$ (iii) Leibniz rule applied to tensor fields $\nabla_X(T(\omega,Y))=(\nabla_XT)(\omega,Y)+T(\nabla_X\omega,Y)+T(\omega,\nabla_XY)$, where $k_1,k_2\in\mathbb{R}$, $f\in C^\infty(M)$, $X,Y$ are vector fields, $\omega\in\Gamma(T^*M)$ is a covector field and $T, S$ are $(1,1)$-tensor fields.

Given a point $p\in M$ and $v\in T_pM$, there is a unique curve $\gamma(\lambda):\mathbb{R}\rightarrow M$ passing $p$ at $\lambda=0$ with velocity $v$ satisfying $\nabla_{\dot\gamma}\dot\gamma=0$, and such curve is termed geodesic denoted by geodesical exponential $\exp_p(\lambda v)=\gamma(\lambda)$, where $\dot\gamma$ is the velocity vector field defined on the curve. Moreover, for an arbitrary curve $\mu:[0,1]\rightarrow M$, a vector field $X$ is called parallel vector field on $\mu$ if $\nabla_{\dot\mu}X=0$. Such vector field $X$ is said to be generated by parallel transport of $X_{\mu(0)}\in T_{\mu(0)}M$. Parallel transport along $\mu$ from $a$ to $b$ is denoted by $\mathcal{P}_{a,\mu}^b$. $\mu$ can be omitted if it's clear from context.

Though $\nabla$ is not a tensor itself, it induces important tensorial objects with fruitful geometric meanings. The Riemannian curvature $\mathcal{R}$ is a $(1,3)$-tensor defined by $\mathcal{R}(X,Y)Z=\nabla_X\nabla_YZ-\nabla_Y\nabla_XZ-\nabla_{[X,Y]}Z$, where $[\cdot,\cdot]$ is the Lie bracket for vector fields. The torsion is a $(1,2)$-tensor denoted by $\tau(X,Y)=\nabla_XY-\nabla_YX-[X,Y]$.

Let $f$ be a transformation on $M$, i.e. a smooth, one-to-one and onto map on $M$. The vector fields $X$ and $X^\prime$ are $f$-related if $(df)_p(X_p)=X^\prime_{f(p)},\forall p\in M$, where $(df)_p:T_pM\rightarrow T_{f(p)}M$ is the push-forward of $f$. $f$ is said to be an affine transformation, a.k.a. connection-preserving map, if for any $f$-related vector field pairs $(X,X^\prime)$ and $(Y,Y^\prime)$, $(df)(\nabla_XY)=\nabla_{X^\prime}Y^\prime$ holds. All affine transformations of $M$ form a Lie group under composition. A local affine transformation $f_U:U\rightarrow M$ on an open subset $U$ is the connection-preserving transformation between $U$ and $f(U)$.

\section{Main Theory}\label{sec::theory}

Our main theory is partitioned into seven parts. We first give a definition of dynamics on manifolds from the vector field perspective. Vector-error representation uses the connection structure on the manifold to rigorously link neighboring flow lines. Linear observed properties are then posed on the flow field from two perspectives: (a) exact linearization with state-estimate-independent Jacobians; (b) preintegration. These two aspects are proven to be equivalent. Existence of linear observed systems is discussed to reveal the extra conditions posed on the curvature of the state space. Parallel curvature or flatness are two sufficient conditions guaranteeing the existence of linear observed systems. The flat Cartan connection case is investigated to show our theory reproduces characterization of group affine systems, being a true generalization. We finally discuss the compatibility with metric, showing the connections encountered in this letter are not necessarily Riemannian.

\subsection{Dynamics on Smooth Manifolds}

A dynamics on a manifold from the vector field point-of-view is defined by assigning a smooth vector field to every vector-valued control input. Details are in \cite{EqF,EqSys,EqFCDC}.

\begin{definition}[Dynamics on Manifold]
    Let $V$ be a finite-dimensional vector space, namely control input. \textbf{A dynamics on a smooth manifold} $M$ is a smooth map $\mathfrak{f}:V\rightarrow\Gamma(TM)$ assigning a smooth vector field $W^{u}$ on $M$, a.k.a. flow field, to every vector-valued control input $u\in V$.
\end{definition}

\begin{definition}[Flow Line]
  \textbf{A flow line of dynamics} $\mathfrak{f}$ on a smooth manifold $M$ is a local one-parameter group of local diffeomorphisms of $W^{u_t}$, denoted $\Phi:(-\epsilon,\epsilon)\times U\rightarrow M,\ (t,p)\mapsto\Phi^t(p)$ with $\epsilon>0$ and $U$ being some open subset of $M$, satisfying:
  \begin{enumerate}
    \item $\forall t\in(-\epsilon,\epsilon)$, $\Phi^t(\cdot):U\rightarrow M$ is smooth;
    \item $\forall p\in U$, $\Phi^t(p):(-\epsilon,\epsilon)\rightarrow M$ in its domain is the integral curve of $W^{u_t}$.
  \end{enumerate}
\end{definition}

So far, the well-definedness of dynamics and its flow line claims no structure beyond smoothness. If the manifold is complete, then the domain of any flow line can be extended to all of $\mathbb{R}$. $\Phi^t(p)$ is then the global one-parameter group of local diffeomorphisms.

\subsection{Vector Error Representation}

We wish to link neighboring flow lines $\Phi^t(p_1)$ and $\Phi^t(p_2)$ of the vector field $W^{u_t}$. A retraction at $p\in M$\cite{Opt_Mfd} allows to move from $p$ to its neighborhood in the direction of a tangent vector $v\in T_pM$. Retractions at different anchor points have no relationship with each other, which are not enough for analyzing the `error-vector' distribution along an entire curve. A fundamental extra structure required besides smoothness is the connection $\nabla$, as introduced in Section~\ref{sec::pre}.

A normal coordinate system $\iota\circ\exp_p^{-1}:U\rightarrow\mathbb{R}^{\dim M}$ at $p\in M$ can be defined by the isomorphism $\iota:T_pM\rightarrow\mathbb{R}^{\dim M}$ identifying $T_pM$ with $\mathbb{R}^{\dim M}$ by virtue of the fact that the $\nabla$-induced $\exp_p$ is a diffeomorphism between $T_pM$ and an open subset $U$ of $M$. Let $x^1,...,x^{\dim M}$ be the normal coordinate system defined above at $p$. Let $U(p,\rho)$ be the neighborhood of $p$ by $\sum_{i=1}^{\dim M}(x^i)^2<\rho^2$. By Theorem 8.7 of Chapter III of \cite{DiffGeo}, there is a positive number $a$, such that for $0<\rho<a$, any two points in $U(p,\rho)$ can be joined by a geodesic which totally lies in $U(p,\rho)$. On such $U$, we can define a smooth error function $\mathfrak{er}:U\times U\rightarrow TM,\ (q,p)\mapsto v$, where $v\in T_pM$ and $\exp_p(v)=q$.

\begin{definition}[Local Vector-error Representation]
  \textbf{A local vector-error representation} of dynamics $\mathfrak{f}$ in an open normal neighborhood $U$ of a geodesically complete manifold $(M,\nabla)$ is a triple $(\Phi^t,\mathfrak{er},u_t)$, where $\Phi^t$, $\mathfrak{er}$, $u_t$ are flow lines of $\mathfrak{f}$, an error function defined on $U\times U$, and a given input sequence respectively.
\end{definition}

\begin{theorem}\label{theorem::a}
    If $p$ and $\hat p$ are in some normal neighborhood $U$ of a complete manifold $M$, where any point pair in $U$ can be joined by a geodesic, then the local vector-error representation $(\Phi^t,\mathfrak{er},u_t)$ induces a smooth vector field $E$ on $\Phi^t(\hat p)$ by $\mathfrak{er}(\Phi^t(p),\Phi^t(\hat p))$ locally in $U$, i.e. $\exp_{\Phi^t(\hat p)}(E_{\Phi^t(\hat p)})=\Phi^t(p),\ \forall t\in I$ and $I:=\left\{t\in\mathbb{R}\vert \Phi^t(p)\in U,\Phi^t(\hat p)\in U\right\}$.
\end{theorem}

The notations used are closely related to observer design. $p$ denotes the true state on $M$ and $\hat p$ is the estimation of the former. $\Phi^t(\hat p)$ denotes the estimated trajectory from $\hat p$ and $\Phi^t(p)$ is viewed as the true trajectory from $p$. $E_{\Phi^t(\hat p)}\in T_{\Phi^t(\hat p)}M$ is termed the error-vector at the estimate $\Phi^t(\hat p)$. This is a true abstraction of the error-state Kalman-type observer whose process model is a copy of the system flow and whose gains are computed by the Jacobians of the error-state ODE\cite{KF_mfd,InEKF}. Completeness is a must such that the error vector is well defined in the entire of each tangent space.

It's worth noting that the most we could expect from the observer designed on manifolds from vector-error representation is local convergence, because geodesical exponential may not be surjective, e.g. the Lie-exponential of $\text{SL}(n,\mathbb{R})$ at the identity.

\subsection{Linear Observed Property from Exact Linearization}

We now start to pose additional `linear observed property' on the flow of a given dynamics. Consider a linear system on the Euclidean space $\mathbb{R}^m$ whose state-space model is $\dot{\bs{x}}=\bs{Ax}$. The flow line emanating from $\bs{x}_0$ can be written as $\Phi^t(\bs{x}_0)=\bs{F}_t\bs{x}_0$, where $\bs{F}_t$ is an operator satisfying matrix ODE $\frac{d}{dt}\bs{F}_t=\bs{A}_t\bs{F}_t$. Any point $\bs{x}=\bs{x_0}+\delta\bs{x}$ in the neighborhood of $\bs{x_0}$ can be reached by moving along a geodesic in the direction of $\delta\bs{x}\in T_{\bs{x}_0}\mathbb{R}^m$. A key observation is that the neighboring flow line can be expressed by $\Phi^t(\bs{x}_0+\delta\bs{x})=\bs{F}_t\delta\bs{x}+\Phi^t(\bs{x}_0)$, where $+$ and $\bs{F}_t$ can be understood as the geodesical exponential and the operator on $\mathbb{R}^m$ preserving Euclidean connection (linearity) respectively. Affine transformation, as a connection-preserving map, thus serves as the core concept in the generalization of exact linearization property to systems on manifolds.

\begin{definition}[Exact Linearization Independent of State]\label{definition::el}
  The local vector-error representation $(\Phi^t,\mathfrak{er},u_t)$ of a system in an open subset $U$ of a manifold $(M,\nabla)$ is said to possess \textbf{exact linearization independent of state} property if:
	\begin{enumerate}
		\item $\forall\hat p\in U$, there exists local affine transformations $\psi^{u_t}:U\rightarrow M$ of $\nabla$, such that any flow line $\Phi^t(p)$ from $p=\exp_{\hat p}(v)$ in the neighborhood of $\hat p$, has the form of $\Phi^t(p)=\exp_{\Phi^t(\hat p)}\circ (d\psi^{u_t})(v)$, where $v\in T_{\hat p}M$;
		\item the linear map $\bs{F}_t:T_{\hat p}M\rightarrow T_{\Phi^t(\hat p)}M$, defined by $\bs{F}_t:=(d\psi^{u_t})_{\hat p}$ with both domain and image chosen to be normal frames, is determined by matrix ODE $\frac{d}{dt}\bs{F}_t=\bs{A}_{u_t}\bs{F}_t$, where $\bs{A}_{u_t}$ is an endomorphism of $T_{\Phi^t(\hat p)}M$ depending only on $u_t$.
	\end{enumerate}
\end{definition}

If an error-state observer is designed for a local vector-error representation with exact linearization independent of state, its process model is a copy of the flow. Unlike general on-manifold EKFs, its error vector ODE is exact without any neglected high-order remainders, leading to theorem~\ref{theorem::b}.
\begin{theorem}\label{theorem::b}
  Consider designing an error-state observer for a system in an open subset $U$ of a manifold $(M,\nabla)$ with exact linearization property on $(\Phi^t,\mathfrak{er},u_t)$. Denote $\hat p\in U$ to be the estimated state and $\hat p\boxplus\bs{e}:=\exp_{\hat p}(\bs{e})$ to be the retraction induced by $\nabla$ for $\bs{e}\in T_{\hat p}M$. The propagation of the error vector defined in the normal frame at the estimated state $\hat{p}_t$ strictly follows linear ODE: $\dot{\bs{e}}_t=\bs{A}_{u_t}\bs{e}_t$. 
\end{theorem}

Exact linearization independent of state is just one aspect of linear observed property. Preintegration is the other side of the coin.

\subsection{Linear Observed Property from Preintegration}

Preintegration is understood as a local self-similarity condition of the flow field, meaning any patch of the flow field is capable of being reconstructed by some transformation of one single neighboring flow line. 

From the vector field point of view, consider a null-homotopy 2-dim patch on a locally simply-connected manifold $M$ and denote $H(t,s):[0,1]\times[0,1]\rightarrow M$ to be the homotopy function. Let $E:=\partial_sH$ be the tangent vectors of a family of curves $H(\cdot,s):[0,1]\rightarrow M,\forall t$ and $T:=\partial_tH$ be the tangent vectors of a family of curves $H(t,\cdot):[0,1]\rightarrow M,\forall s$. The following conditions, denoted `condition (*)', can be posed on the patch, as in Fig.~\ref{fig::patch_condition}:
\begin{enumerate}
	\item $H(t,0),H(t,1)$ are $\Phi^t(p_1),\Phi^t(p_2)$ respectively;
	\item $E$ restricted on $\Phi^t(p_2)$ is the error-vector w.r.t. $\Phi^t(p_1)$;
  \item For $t\in[0,1]$, the family of geodesics $s\in[0,1]\mapsto\exp_{\Phi^t(p_2)}(s E_{{\Phi^t(p_2)}})$ happens to be $H(\cdot,s)$ parameterized by $s$. Moreover, $E$ could be extended to the entire patch by parallelly transported along the corresponding geodesics, i.e. $E_{H(s,t)}$ is the parallel transport of $E_{\Phi^t(p_2)}$ along the geodesic $H(\cdot, s)$ parameterized by $s$. 
\end{enumerate}

\begin{figure}[t]
  \centering
  \parbox{3in}{
    \centering
    \includegraphics[scale=0.60]{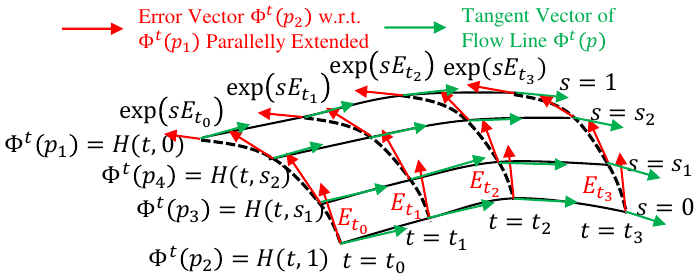}
  }
  \caption{An intuitive picture of a self-similar flow field patch and its related vector fields $E$ and $T$.}
  \label{fig::patch_condition}
\end{figure}

\begin{definition}[Self-Similar Flow Field]
  The flow field of $(\Phi^t,\mathfrak{er},u_t)$ is \textbf{self-similar} if on any patch with condition $(*)$, the intermediate curve $H(t,\cdot):[0,1]\rightarrow M$ parameterized by $t$ is again a flow line for any $s$, as in Fig.~\ref{fig::patch_condition}. 
\end{definition}

\begin{theorem}\label{theorem::c}
  Given a self-similar $(\Phi^t,\mathfrak{er},u_t)$ on a manifold $(M,\nabla)$ with parallel torsion, the tangent vectors $T$ of a family of flow lines coincide with the Jacobi field restricted to a certain geodesic at a fixed timestamp $t$.
\end{theorem}

The above theorem implies the velocity vectors of neighboring flow lines can be calculated by the transformation determined by Jacobi field equation, which is an initial value problem of the corresponding ODE with respect to parameter $s$ on a geodesic at a fixed timestamp $t$.

\begin{definition}[Preintegration]
  A local vector-error representation $(\Phi^t,\mathfrak{er},u_t)$ in an open subset $U$ of a manifold $(M,\nabla)$ with parallel torsion is \textbf{preintegrable} if its flow field is self-similar and for any neighboring flow lines $\Phi^t(p_1)$ and $\Phi^t(p_2)$ in the patch with condition $(*)$, the transformation from $\Phi^t(p_2)$ to $\Phi^t(p_1)$, i.e. the vector transport by the initial-value problem of the Jacobi field ODE, is a pre-calculated transformation of any starting point on $\Phi^t(p_2)$. 
\end{definition}

\begin{remark}\label{remark::preint}
  The pre-calculated transformation defined by the descriptive languages above is independent of $t$ and maps the tangent vector $T$ from $T_{\Phi^t(p_2)}M$ to $T_{\Phi^t(p_1)}M$ along the geodesic $\gamma(s):[0,1]\rightarrow M$ connecting $p_2$ to $p_1$. By self-similarity, the homotopy function $H$ of a preintegrable system must have the form of $H(t,s)=\varphi_t\circ\gamma(s)$.
\end{remark}

This is indeed the concept of `preintegration' in the field of robotics\cite{Pre_Int,PreIntMath}, since such transformation from a pre-chosen point $p_0$ along the geodesic $\gamma(s)$ from $p_0$ is calculated just once for all initial states in the neighborhood of $p_0$ in $M$.
\begin{theorem}\label{theorem::preint_property}
  Any flow line $\Phi^t(p)$ of a preintegrable system on a manifold $M$ with connection $\nabla$ coincides with a local affine transformation $\varphi_t$ of $(M,\nabla)$ at any $t$.
\end{theorem}

\subsection{Existence of Linear Observed Systems}

So far, we have rigorous descriptive properties of preintegration and exact linearization with state independent Jacobians. Definition of linear observed systems on arbitrary smooth manifolds is made possible by characterization from these two aspects. Preintegration and exact linearization with state-estimate-independent Jacobians are two equivalent concepts, similar to the case when the state space is a Lie group, as discussed in \cite{LoS}.
\begin{theorem}\label{theorem::el_is_preint}
  A system with local vector-error representation $(\Phi^t,\mathfrak{er},u_t)$ on a manifold $(M,\nabla)$ with parallel torsion possesses exact linearization with state-estimate-independent Jacobian property if and only if it's preintegrable. 
\end{theorem}

\begin{definition}[Linear Observed Systems]
  A system on a complete manifold $(M,\nabla)$ is \textbf{linear observed} if any local vector-error representation is preintegrable or possesses exact linearization property.
\end{definition}

At a first glance, existence of linear observed systems is simply the existence of affine transformation groups on $(M,\nabla)$. However, this is not enough, since the property of state-independence in both exact linearization and preintegrability requires the corresponding local affine transformation $\psi$ is capable of being reconstructed from a single linear isomorphism between a pair of points in $M$. Consider the linear isomorphism $\bs{F}_t:T_{\hat p}M\rightarrow T_{\Phi^t(\hat p)}M$ in definition~\ref{definition::el}. Establishment of linear observed property requires the ability of $\bs{F}_t$ being uniquely extended to a local affine transformation on an open set containing $\hat p$, i.e. there exists a unique local affine transformation $\psi$, such that $d\psi$ at $T_{\hat p} M$ is exactly $\bs{F}_t$. This has an effect of posing restrictions on the state space $M$ itself beyond requirements of the specific forms that the dynamics should obey.
\begin{assumption}
  The state manifold $M$ is real-analytic.
\end{assumption}
\begin{theorem}\label{theorem::d}
  On a complete real-analytic manifold $(M,\nabla)$ with parallel torsion, there exists linear observed systems if and only if the linear isomorphisms $\bs{F}_t$ in definition~\ref{definition::el} maps the tensor $\nabla^m\mathcal{R}$ into $\nabla^m{R^\prime},\forall m\in\mathbb{N}$, where $(\mathcal{R},\mathcal{R}^\prime)$ are $\bs{F}_t$-related curvature tensors.
\end{theorem}
\begin{corollary}\label{corollary::parallel_zero_curvature}
  On a complete analytic manifold $(M,\nabla)$ with parallel torsion, there exists linear observed systems if the curvature tensor $\mathcal{R}$ is parallel $\nabla\mathcal{R}=0$ or zero $\mathcal{R}=0$.
\end{corollary}

\subsection{$\mathcal{R}$=0 Classifies Linear Observed Systems on Groups}

Flatness has non-trivial consequences on the structure of the state space $M$ under mild topological assumptions. The concepts of linear observed property and flatness are local, and thus there must be additional topological conditions to get global results.

\begin{assumption}\label{assumption::c}
  The state space $M$ is simply-connected, i.e. every loop in $M$ is homotopic with a constant loop.
\end{assumption}

Though some state manifolds of our interest are not eligible to be simply-connected (e.g. ${\rm SO}(3)$), it's always possible to construct a corresponding simply-connected covering space (e.g. ${\rm SU}(2)$ as a double cover of ${\rm SO}(3)$) and lift the dynamics to such covering space. In general, there always exists a simply-connected covering space $\tilde M$ for connected state manifold $M$. The flow $\Phi^t$ of a system on $M$, as a local one-parameter group of local diffeomorphisms, locally induces a lifted object $\pi^{-1}\circ\Phi^t\circ\pi$ on $\tilde M$, easily seen to be a flow on $\tilde M$ where $\pi:\tilde M\rightarrow M$ is the projection map of the lift as a local diffeomorphism. This implies the assumption of simply-connectedness is not considered a limitation.

\begin{theorem}\label{theorem::e}
  There exists linear observed systems on a complete simply-connected manifold $(M,\nabla)$ whose torsion is parallel and curvature is zero. The manifold moreover admits a Lie group structure whose left multiplication is a parallel transport, and the corresponding $\nabla$ is the Cartan (+)-connection. The local affine transformation encountered in the definition of linear observed systems can be extended to all of $M$ and is a group-automorphism under the multiplication induced by $\nabla$.
\end{theorem}
\begin{remark}
  Let $\tilde X,\tilde Y,\tilde Z$ be left-invariant vector fields on a Lie group $G$. A family of connections $\nabla^\mu$ parameterized by $\mu\in[0,1]$, whose geodesics coincide with the one-parameter groups generated by the corresponding Lie algebra, have such properties: (a) $\nabla^\mu_{\tilde X}{\tilde Y}=\mu[\tilde X,\tilde Y]$ (b) $\mathcal{R}(\tilde X,\tilde Y)\tilde Z=\mu(\mu-1)[[\tilde X,\tilde Y],\tilde Z]$ (c) $\tau(\tilde X,\tilde Y)=(2\mu-1)[\tilde X,\tilde Y]$. We herein consider the (+)-connection with $\mu=1$. 
\end{remark}

As an application of theorem~\ref{theorem::e}, we reproduce the characterization of linear observed systems on Lie groups (group affine systems) proposed in the work of Bonnabel\cite{InEKF,LoS} from the perspective of the Cartan (+)-connection, while the motivation in \cite{LoS} is error autonomy.

\begin{corollary}\label{corollary::b}
  Given state space as a Lie group $G$, consider a linear observed system whose local vector-error representation is $(\Phi^t,\mathfrak{er},u_t)$, then the flow lines satisfy:
	\begin{equation*}
	  \Phi^t(g\cdot h)=\psi^t(g)\cdot\Phi^t(h),\quad\forall g,h\in G
	\end{equation*}
	where $\psi^t:G\rightarrow G$ must be a group-automorphism on $G$.
\end{corollary}

\subsection{Discussions}

Readers may wonder if the connection encountered in the previous subsections is compatible with a certain metric on the manifold. This fails when the metric on Lie group $G$ is not bi-invariant. It's equivalent to requiring the group being isomorphic to the direct product of a compact group and some $\mathbb{R}^m$\cite{DGLG}, which is not true for ${\rm SE}(3)$ or ${\rm SE}_2(3)$ in robotics. Moreover, the torsion tensor on left-invariant vector fields coincides with the Lie bracket under the circumstance of theorem~\ref{theorem::e}, being non-zero if the group is non-abelian. 

\section{An Example from Navigation}\label{sec::example}

An example of the attitude indicator on aircrafts, i.e. partial attitude estimator by IMU, serves to unfold the mechanisms of all the above theorems. Denote $\bs{R}\in\text{SO}(3)$ to be the rotation of IMU from its body frame to the world frame. The process ODE driven by the angular velocity $\bs{\omega}\in\mathbb{R}^3$ and the algebraic measurement are modelled in (\ref{eq::full_state_ahrs}),
\begin{equation}
  \label{eq::full_state_ahrs}
  \dot{\bs{R}}=\bs{R}\bs{\omega}^\times,\quad \bs{a}=\bs{R}^{\rm T}\bs{g}
\end{equation}
where $(\cdot)^\times$ embeds an $\mathbb{R}^3$ into skew-symmetric $\mathbb{R}^{3\times 3}$ matrix. $\bs{a},\bs{g}$ are accelerometer observation and gravity respectively. Since yaw (rotation along gravity) forms an unobservable subspace $\text{SO}(2)$, we eliminate the yaw unobservability by estimating on the quotient space $\text{SO}(3)/\text{SO}(2)$ diffeomorphic to $S^2$ encoding the remaining roll and pitch. Specifically, let $\bs{q}=\bs{R}^{\rm T}\bs{g}\in S^2\subset\mathbb{R}^3$ be the new state of our estimator, and the fully-observable system equations are now (\ref{eq::partial_state_ahrs}).
\begin{equation}
  \label{eq::partial_state_ahrs}
  \dot{\bs{q}}=-\bs{\omega}^\times\bs{q},\quad \bs{a}=\bs{q}
\end{equation}
Note the state evolves strictly on $S^2$, canonically embedded in the ambient Euclidean space $\mathbb{R}^3$ to avoid parameterization by local coordinates. The system is strictly linear from its embedding in $\mathbb{R}^3$, but investigating its linear observed property intrinsically on the sphere is non-trivial. We intentionally choose this non-Lie-group-type state to reveal the mechanisms of our theory to generalize linear observed property beyond systems on Lie groups and distinguish our work from existing attitude estimators on groups.

Under the usual round connection $\nabla$, there exists linear observed systems on $S^2$ by corollary~\ref{corollary::parallel_zero_curvature} due to its constant curvature nature implying $\nabla\mathcal{R}_{S^2}=0$.
\begin{theorem}\label{theorem::example_a}
  The system (\ref{eq::partial_state_ahrs}) on $S^2$ is linear observed under $\nabla$ induced by the canonical embedding of $S^2$ into $\mathbb{R}^3$. Moreover, all linear observed systems on $S^2$ with such $\nabla$ are classified into the form of $\dot{\bs{q}}=\bs{\sigma}_{u_t}^\times\bs{q}$, where $\bs{q}\in S^2\subset\mathbb{R}^3$ and $\bs{\sigma}_{u_t}:V\rightarrow\mathbb{R}^3$ is an arbitrarily-chosen function of the control input $u_t\in V$.
\end{theorem}

Decompose the state $\bs{q}$ into a nominal estimate $\hat{\bs{q}}$ and a vector error $\delta\bs{q}\in T_{\hat{\bs{q}}}S^2$ as (\ref{eq::update}), where $\Pi_{\hat{\bs{q}}}:=\bs{I}-\hat{\bs{q}}\hat{\bs{q}}^\top/\Vert\hat{\bs{q}}\Vert^2$.
\begin{equation}
  \label{eq::update}
  \bs{q}=\exp_{\hat{\bs{q}}}(\delta\bs{q})=\hat{\bs{q}}+\Pi_{\hat{\bs{q}}}\delta\bs{q}
\end{equation}
Treating $\bs{a}$ as on $S^2$, output error $\delta\bs{a}\in T_{\hat{\bs{a}}}S^2$ can be defined.
\begin{align}
\label{eq::ekf1}
\dot{\hat{\bs{q}}}&=-\bs{\omega}^\times\hat{\bs{q}},\quad \hat{\bs{a}}=\hat{\bs{q}}\\
\label{eq::ekf2}
\delta\dot{\bs{q}}&=-\bs{\omega}^\times\delta\bs{q},\quad \delta\bs{a}=\delta\bs{q}
\end{align}
Associating noise to (\ref{eq::ekf1}-\ref{eq::ekf2}), a stable EKF-like observer can be designed with nonlinear update (\ref{eq::update}) by \cite{InEKF}.

\section{Conclusions}\label{sec::conclusion}

In this letter, we establish linear observed properties for dynamics evolving on arbitrary smooth manifolds by leveraging the affine connection structure endowed upon the state space. The key results are linear observed properties pose conditions on the curvature of the state space, beyond forcing the flow of the dynamics to be compatible with some local connection-preserving transformations. Specifically, the Cartan flat connection case reconstructs the classification of linear observed systems on Lie groups. We show that our theory is a true generalization of linear observed property in\cite{LoS} by providing a simple example from navigation with non-Lie-group-type state space.

\section*{APPENDIX}

\subsection{Proof of Theorem~\ref{theorem::a}}

\begin{proof}
  On a manifold $(M,\nabla)$, around any point there exists a normal neighborhood $U$ such that any two points in $U$ can be joined by a geodesic which totally lies in $U$. $\Phi^t(p)$ and $\Phi^t(\hat p)$ are all in $U$ for $t\in I$, and thus $\mathfrak{er}$ can be applied to induce the error vector field.
\end{proof}

\subsection{Proof of Theorem~\ref{theorem::b}}

\begin{proof}
    Denote $\Phi^t(\hat p)$ to be the estimated trajectory emanating from $\hat p\in M$ and $\Phi^t(p)$ to be the true trajectory of the state. Let $\bs{e}_t\in T_{\Phi^t(\hat p)}M$ be the error vector at time $t$. By the definition of the retraction, $\exp_{\Phi^t(\hat p)}(\bs{e}_t)=\Phi^t(p)$. Exact linearization at $\Phi^t(\hat p)$ guarantees $\bs{e}_t=(d\psi^{u_t})_{\hat p}(\bs{e}_0)$, where $\bs{e}_0\in T_{\hat p}M$ is the error vector at $\hat p$. Differentiating $\bs{e}_t$ yields $\frac{d}{dt}\bs{e}_t=[\frac{d}{dt}\bs{F}_t]\bs{e}_0=\bs{A}_{u_t}\bs{F}_{u_t}\bs{e}_0=\bs{A}_{u_t}\bs{e}_t$.
\end{proof}

\subsection{Proof of Theorem~\ref{theorem::c}}

\begin{lemma}\label{lemma::a}
  $E$ and $T$ defined in condition $(*)$ commute, i.e. $[E,T]=0$.
\end{lemma}
\begin{proof}
  For any smooth function $\tilde g\in C^\infty(M)$ restricted on the patch, $[E,T]=0$ follows from the commutativity of partial derivatives $\partial_s\partial_t(\tilde g\circ H)=\partial_t\partial_s(\tilde g\circ H)$. 
\end{proof}

{\noindent\hspace{2em}{\itshape Proof of Theorem~\ref{theorem::c}: }}For a fixed timestamp $t$, consider points on two neighboring flow lines $\Phi^t(p_1)$ and $\Phi^t(p_2)$. Let the geodesic emanating from $\Phi^t(p_2)$ hitting $\Phi^t(p_1)$ be $H(\cdot, s)$ parameterized by $s$. Let these two flow lines be the pair of edges of the patch satisfying condition $(*)$. By self-similarity, $T=\partial_tH(t,s)$ is the family of velocity tangent vectors of some flow line restricted on the former geodesic parameterized by $s$. By parallel torsion and lemma~\ref{lemma::a}, we can establish $\nabla_{\cdot}\nabla_ET=\nabla_{\cdot}\nabla_TE$ from $\nabla_{\cdot}(\nabla_ET-\nabla_TE-[E,T])=0$. Direct calculation then yields:
\begin{align*}
  \nabla_E\nabla_ET&=\nabla_E\nabla_TE-0-0\\
  &=\nabla_E\nabla_TE-\nabla_T(\nabla_EE)-\nabla_{[E,T]}E\\
  &=\mathcal{R}(E,T)E=-\mathcal{R}(T,E)E
\end{align*}
which is exactly the Jacobi field equation on $H(\cdot, s)$.
\endproof

\subsection{Proof of Theorem~\ref{theorem::preint_property}}

\begin{proof}
    From the definition of preintegrability, an arbitrary flow line $\Phi^t(p)$ of the system is expressed as a curve parameterized by $t$ of the patch homotopy function $H(t,s):[0,1]\times[0,1]\rightarrow M$ with the form $H(t,s)=\varphi_t\circ\gamma(s)$ defined in Remark~\ref{remark::preint}, where $\varphi_t$ is a local one-parameter group of local diffeomorphisms on $M$ and $\gamma_s$ is a geodesical segment. Clearly, $\Phi^t(p)$ coincides with $\varphi^t$ in the domain of the latter. Note that the choice of $H(t,s)$ and thus $\gamma(s)$ is quite arbitrary. Condition $(*)$ further guarantees that $\partial_sH(t,s)$ is again a geodesic, making $\varphi_t$ a geodesic-preserving map at least in its domain. As a result, $\Phi^t(p)$, or $\varphi_t$ equivalently, happens to be the connection-preserving map in its domain, i.e. a local affine transformation.
\end{proof}

\subsection{Proof of Theorem~\ref{theorem::el_is_preint}}

\begin{proof}
  Let $\Phi^t$ be an arbitrary flow line of a system on a manifold $(M,\nabla)$ with parallel torsion and $\psi:U\rightarrow M$ be a local affine transformation. The defining properties of $\psi$ imply it maps a geodesical segment in $U$ to a curve segment in $f(U)$ which is again a geodesic, and vice versa. The commutativity of $\psi$ and $\exp$ yields
  \begin{equation*}
  \Phi^t(p)=\exp_{\Phi^t(\hat p)}\circ(d\psi^{u_t})(v)=\psi^{u_t}(\exp_{\hat p}(v))
  \end{equation*}. This can be interpreted as the equivalence of exact linearization and preintegrabitilty by theorem~\ref{theorem::preint_property}.
\end{proof}

\subsection{Proof of Theorem~\ref{theorem::d} and Its Corollary}

\begin{proof}
  Let $(X,Y,Z)$ and $(X^\prime,Y^\prime,Z^\prime)$ be $\bs{F}_t$-related vector fields defined in an open subset of $M$ containing the anchor point of $\bs{F}_t$. Denote $\mathcal{R}$ and $\tau$ to be the induced curvature and torsion tensors by $\nabla$. The conditions in this theorem are equivalent to requiring $\nabla^m\mathcal{R}(X,Y)Z$ and $\nabla^m\tau(X,Y)$ being $\bs{F}_t$-related to $\nabla^m\mathcal{R}(X^\prime,Y^\prime)Z^\prime$ and $\nabla^m\tau(X^\prime,Y^\prime)$ respectively. By theorem 7.2 of chapter VI of \cite{DiffGeo}, there exists a unique local affine transformation $\varphi_t$ in the neighborhood of the anchor point of $\bs{F}_t$, such that $d\varphi^t$ coincides with $\bs{F}_t$ at its anchor point. This has fulfilled the extra requirements for the existence of linear observed systems discussed before this theorem. $\nabla\mathcal{R}=0$ or $\mathcal{R}=0$ clearly satisfies the condition, proving the corollary.
\end{proof}

\subsection{Proof of Theorem~\ref{theorem::e} and Its Corollary}

\begin{lemma}\label{lemma::c}
  Consider a simply-connected complete manifold $(M,\nabla)$ with flat $\nabla$ and parallel torsion, then $M$ admits a Lie group structure whose parallel transport coincides with the left translation induced by group multiplication\cite{TheoConn}.
\end{lemma}
\begin{proof}
  Proved by N. Hicks in \cite{TheoConn}.
\end{proof}

{\noindent\hspace{2em}{\itshape Proof of Theorem~\ref{theorem::e}: }}The proof of the first half follows directly by combining corollary~\ref{corollary::parallel_zero_curvature} and lemma~\ref{lemma::c}. Since Lie group is analytic, the local affine transformation in the definition of linear observed property can be uniquely extended to the entire of the manifold $M$ with simply-connectedness. The affine transformation preserves $\nabla$, and thus preserves the multiplication induced by $\nabla$, making it an automorphism.
\endproof

{\noindent\hspace{2em}{\itshape Proof of Corollary~\ref{corollary::b}: }}A direct consequence of theorem~\ref{theorem::e} and the definition of linear observed property by $\nabla$.
\endproof

\subsection{Proof of Theorem~\ref{theorem::example_a}}

\begin{proof}
  The $\nabla$ herein is compatible with the round metric on $S^2$, implying the affine transformation group coincides with the isometry group of $S^2$. The largest isometry group of $S^2$ is ${\rm O}(3)$. By the definition of exact linearization, the linear isomorphism $\bs{F}_t:T_{p_0}S^2\rightarrow T_{q_0}S^2$ evolves on ${\rm O}(3)$. Since $\frac{d}{dt}\bs{F}_t=\bs{A}_{u_t}\bs{F}_t$, we conclude $\bs{A}_{u_t}$, as a function of the input $u_t$, evolves in the space of $\dot{\bs{F}}_t\bs{F}_t^\top$, i.e. the Lie algebra of ${\rm O}(3)$. It follows that all linear observed systems on $S^2$ with this $\nabla$ are in the form of $\dot{\bs{q}}=\bs{\sigma}^\times_{u_t}\bs{q}$ including (\ref{eq::partial_state_ahrs}), where $\bs{\sigma}_{u_t}$ is a smooth function of $V\rightarrow\mathbb{R}^3$.
\end{proof}







\bibliographystyle{IEEEtran.bst}
\bibliography{refs.bib}

\end{document}